\documentclass[a4paper,12pt,fleqn]{article}
\usepackage[latin1]{inputenc}  \usepackage[T1]{fontenc}
\usepackage{lscape}
\usepackage{tabularx}     
\usepackage{pstricks}    
\usepackage{rotating}
\usepackage{amsfonts,amsmath,amsthm,amssymb,dsfont}
\usepackage{marginnote}
\usepackage{rotate}
\usepackage{pgfplots}
\usepackage{tikz}
\usepackage{tikz-cd}
\usepackage[arrow,matrix,curve]{xy} 
\usepackage{todonotes}	
\title{An Aspect of Optimal Regression Design for LSMC}
\author{Christian Wei\ss{} \and Zoran Nikoli\'{c}}
\date{\today}

\newtheorem{thm}{Theorem}[section]

\newtheorem{lem}[thm]{Lemma}

\newcommand{\RR}{{\mathbb{R}}}

\newcommand{\NN}{{\mathbb{N}}}

\makeindex

\begin{document} 

\maketitle

\begin{abstract}% englische Fassung
	Practitioners sometimes suggest to use a combination of Sobol sequences and orthonormal polynomials when applying an LSMC algorithm for evaluation of option prices or in the context of risk capital calculation under the Solvency II regime. In this paper, we give a theoretical justification why good implementations of an LSMC algorithm should indeed combine these two features in order to assure numerical stability. Moreover, an explicit bound for the number of outer scenarios necessary to guarantee a prescribed degree of numerical stability is derived. We embed our observations into a coherent presentation of the theoretical background of LSMC in the insurance setting.
\end{abstract}

\section{Introduction}

Least Squares Monte Carlo (LSMC) methods were originally introduced as an alternative to classical Monte Carlo approaches when calculating the price of an American or Bermuda style option, for which no closed-form solutions exist, compare e.g. \cite{Car96}, \cite{LS01}. In recent years LSMC has gained a lot of attention also in the insurance business, where approximation algorithms are needed to calculate the capital requirements under the Solvency II regime, see e.g. \cite{BBR10}, \cite{LHKB13}, \cite{BFW14}, \cite{KNK18}. The reason for the necessity of approximation is that a full nested stochastic calculation of the capital requirement would cause run times which as of today by far exceed the computing capacities of insurance companies.\\[12pt]
A theoretical justification for applying an LSMC approach in the insurance context was given by \cite{BH15}, who, using a result of Newey \cite{New97}, formally proved  convergence of the algorithm for the risk distribution and for certain families of risk measures.\footnote{Actually, they also discuss that the estimator for the Value-at-Risk, which is demanded under the Solvency II regime, is biased, but on the conservative side, compare also \cite{GJ10}.} Convergence under less restrictive assumptions than those from \cite{New97} was proved in \cite{Ben17}. It is stated in \cite{KNK18} that convergence also holds in the slightly different setting closer to the actual implementations on the market in contrast to \cite{BH15} which we are going to present here.\\[12pt] 
Let us shortly describe how the LSMC algorithm works: As a first step, risk drivers $Z_1,\ldots,Z_s$ which are relevant for the insurance company are identified, among them market and underwriting risks. In the practical implementations typically each risk driver is confined to a compact range, e.g. given by the $0.1$ to the $99.9$ percentile of the real-world distribution of this risk driver. Thus, we may without loss of generalization assume that $(Z_1,\ldots,Z_s) \in [0,1]^s$ after scaling. Next, a fitting space is constructed by deterministically choosing many (usually several thousand) realizations of $Z(\omega_i):=(Z_1(\omega_i),\ldots,Z_s(\omega_i))$. Often, Sobol sequences, a special type of low-discrepancy sequences, are chosen at this step to uniformly fill $[0,1]^s$. These so-called \textbf{outer scenarios} are fed into the cashflow projection model of the insurance company and the best estimate liabilities (BEL) is evaluated for a small number (e.g. 1 or 2) of so-called \textbf{inner scenarios}, i.e. Monte-Carlo simulations under the risk-neutral measure conditioned on the risk driver realization (outer scenario) under consideration. Afterwards a regression is performed: the BEL values are taken as response $y$. The design matrix $X$ consists of the basis functions $\varphi_j$ evaluated at the risk driver realizations $Z(\omega_i)$, that is $X_{ij} = \varphi_j(Z(\omega_i))$. The regression problem thus takes the form
\begin{align} \label{eq:regression}
y = X\beta + \epsilon,
\end{align}
where the parameter vector $\beta$ needs to be estimated and $\epsilon$ denotes the error term. Usually, a least squares estimation for $\beta$ is applied, however alternative types of regression were shown to be efficient as well, see \cite{NJZ17}. Finally, the quality of the approximation is assessed using an out-of-sample validation. For more details on the LSMC algorithm we refer the reader to \cite{BFW14} and \cite{KNK18}.\\[12pt]
In this paper we are interested in the following observation: At first sight, it seems odd to fill the whole space $[0,1]^s$ uniformly by a Sobol sequence: if the risk drivers were uncorrelated the \textit{corners} of the unit cube correspond to the $0.01^s$ respectively $(1-0.01)^s$ percentile of the joint distribution.\footnote{In practice, Gauss copulas are widely used to model dependencies of the risk drivers. This results in a concentration inside the unit sphere.} However, we will argue here that for numerical reasons this is indeed the best way to perform the LSMC algorithm: numerical stability can be achieved optimally by combining a low-discrepancy sequence with a (subset of an) orthonormal polynomial basis of the Hilbert space $L^2([0,1]^s)$ as basis functions for the regression.

\section{Theoretical Background}

\paragraph{The numerical challenge.} While the most time consuming step in the LSMC calculation is the evaluation of the cashflow projection model, the numerically most challenging step is the regression. The $N$-dimensional estimator $\widehat{\beta}$ of the parameter vector in the linear regression \eqref{eq:regression} is given by
\begin{align} \label{eq:estimator}
\widehat{\beta} = (X^TX)^{-1}X^Ty = \left(\frac{1}{N} X^TX\right)^{-1} \cdot \frac{1}{N} X^Ty.
\end{align}
While matrix multiplication is numerically stable, the main problem lies here in the inversion of $X^TX$ because the number of columns of $X$ might be big (equal to the dimension $N$ of the space which the regression projects on). The matrix $X^TX$ might be ill-conditioned. This has led to the implementation of various regularization techniques, the most famous being probably ridge regression, see e.g. \cite{TGSY95} and \cite{NJZ17}. Our approach to the problem is to add the multiplicative factor $\frac{1}{N}$ twice: as we will prove in this paper, it will first stabilize the inversion of the matrix. Second the values of $y$ are in the context of LSMC only based on a small number of inner scenarios (as stated earlier $<10$) and hence they are very inaccurate. If different inner scenarios (Monte Carlo simulations under the risk-neutral measure) were evaluated the response $y$ would thus differ a lot. It would on the other hand be desirable that the estimation in \eqref{eq:estimator} yields a similar estimated parameter vector $\widehat{\beta}$ in either simulation. Therefore, it makes sense to add the factor $\tfrac{1}{N}$ which scales down the inaccuracies. 
\paragraph{Condition number.} Recall that the \textbf{condition number} $\kappa(A)$ of a matrix $A$ measures the numerical stability of a matrix, i.e. it gives a bound how inaccurate the solution of the linear equation $Ax=b$ is. It is defined as
$$\kappa(A) = \left\| A^{-1} \right\| \cdot \left\| A \right\|,$$
where $\left\| \cdot \right\|$ is the $l^2$-operator norm. If small changes in the input result in large changes in the output, then the matrix is called \textbf{ill-conditioned}, otherwise \textbf{well-conditioned}. Since the matrix $A=X^T X$ is a normal matrix we have 
\begin{align} \label{eq:kappa2}
 \kappa(A) = \frac{\lambda_{\max}(A)}{\lambda_{\min}(A)},
\end{align}
where $\lambda_{\max}(A)$ and $\lambda_{\min}(A)$ denote the largest and smallest eigenvalue of $A$. Often $\widetilde{\kappa} = \log(\kappa)$ is considered instead of $\kappa$ because it can be interpreted as the number of (last) digits which may be incorrect due to numerical instability of the regression problem. For more details we refer the reader e.g. to \cite{TB97}, Chapter III.
\paragraph{Gershgorin Circle Theorem} A main ingredient in our proof is Gershgorin's Theorem which gives a bound for the eigenvalues of a matrix. It was first proved in \cite{Ger31} and belongs to the clasical results of numerical analysis. We state it here for the sake of completeness.
\begin{thm}[Gershgorin Circle Theorem] \label{thm:Gershgorin} All eigenvalues of a matrix $A \in \mathbb{C}^{n \times n}$ lie within the Gershgorin discs
	$$D_j := \left\{ z \in \mathbb{C} \, | \, |z - a_{ii}| \leq \sum_{k=1, k \neq j}^n |a_{jk}| \right\}.$$
\end{thm}

\paragraph{Orthonormal polynomials.} The space of square-integrable functions $L^2([0,1]^s)$, becomes a Hilbert space when equipped with the inner product
$$\langle f,g \rangle := \int_{[0,1]^s} f(x)g(x) \mathrm{d} x.$$
The Hilbert norm of an element $f \in  L^2([0,1]^s)$ is given by $\left\| f \right\| = \sqrt{\langle f,f \rangle}.$ A subset $S \subset L^2([0,1]^s)$ is \textbf{orthogonal} if $\langle f,g\rangle = 0$ for every two elements $f,g \in S$. If in addition $\left\| f \right\| = 1$ holds for all $f \in S$, then $S$ is \textbf{orthonormal} and called a \textbf{(Hilbert) basis}. For a complete basis $S$, we can write every element $x \in L^2([0,1]^s)$ as
$$x = \sum_{u \in S} \langle x,u \rangle u,$$
i.e. every element can be arbitrarily well approximated by linear combinations of basis elements. Note that $L^2([0,1]^s)$ is a separable Hilbert space, so that complete bases are available. To explicitly construct a Hilbert basis e.g. the following steps can be applied: For dimension $s=1$ one ends up with Legendre-like polynomials $P_n(x)$ when starting with monomials and applying Gram-Schmidt algorithm.\footnote{There are also other and more involved examples like the Askey-Wilson polynomials introduced in \cite{AW85}.} For any one-dimensional orthonormal basis $p_1(x), p_2(x),\ldots$ of $L^2([0,1])$, an $s$-dimensional Hilbert basis of $L^2([0,1]^s)$ can be obtained as follows:
\begin{lem} Let $p_1(x), p_2(x),\ldots$ be an orthonormal basis of $L^2([0,1])$. Then the multi-dimensional elements
\begin{align} \label{eq:pols}
p_{i_1,i_2,\ldots,i_s}(x_1,\ldots,x_s) := p_{i_1}(x_1) \cdot p_{i_2}(x_2) \cdot \ldots \cdot p_{i_s}(x_s)
\end{align}
with $i_j \in \NN$ for $1 \leq j \leq s$ form a basis of $L^2([0,1]^s)$.
\end{lem}
\begin{proof}
 By Fubini's Theorem and the property that the $p_i$ are a basis in $L^2([0,1])$ it follows that
 \begin{align*}
  \int_0^1 \int_0^1 & \ldots \int_0^1  (p_{i_1}(x_1) \cdot p_{i_2}(x_2) \cdot  \ldots \cdot p_{i_s}(x_s))^2 \mathrm{d}x_1 \mathrm{d}x_2 \ldots \mathrm{d}x_s\\
  & =  \int_0^1 \int_0^1 \ldots \int_0^1 p_{i_1}(x_1)^2 \mathrm{d}x_1 \cdot p_{i_2}(x_2)^2 \mathrm{d}x_2 \cdot \ldots \cdot p_{i_s}(x_s)^2   \ldots \mathrm{d}x_s\\
  & = 1 \cdot 1 \cdot \ldots \cdot 1 = 1.
 \end{align*}
 Similarly, for two polynomials $p_{i_1,i_2,\ldots,i_s}$ and $p_{j_1,j_2,\ldots,j_s}$ with $i_k \neq j_k$ for some $k$ we get
  \begin{align*}
 \int_0^1 & \ldots \int_0^1  p_{i_1}(x_1) \cdot \ldots \cdot p_{i_s}(x_s) p_{j_1}(x_1) \cdot \ldots \cdot p_{j_s}(x_s) \mathrm{d}x_1 \ldots \mathrm{d}x_s\\
 & =  \int_0^1 \int_0^1 \ldots \int_0^1 p_{i_k}(x_k) p_{j_k}(x_k) \mathrm{d}x_k p_{i_1}(x_1) p_{j_1}(x_1) \mathrm{d}x_1 \ldots \mathrm{d}x_s\\
 & = 0 \cdot \ldots = 0.
 \end{align*}
\end{proof}
\paragraph{Discrepancy.} Let $Z=(z_n)_{n \geq 0}$ be a sequence in $[0,1)^s$. Recall that the \textbf{star-discrepancy} of the first $N$ points of the sequence is defined by
$$D^*_N(Z) := \sup_{B \subset [0,1)^d} \left| \frac{A_N(B)}{N} - \lambda_s(B) \right|,$$
where the supremum is taken over all intervals $B = [0,a_1) \times [0,a_2) \times \ldots \times [0,a_s) \subset [0,1)^s$ and $A_N(B) :=  |\left\{ n \ \mid \ 0 \leq n < N, z_n \in B \right\}|$ and $\lambda_s$ denotes the $s$-dimensional Lebesgue-measure. If $D^*_N(Z)$ satisfies
\begin{align} \label{eq:diskr}
D^*_N(Z) = O(N^{-1}(\log N)^{s})
\end{align}
then $Z$ is called a \textbf{low-discrepancy sequence}. It is widely conjectured that the rate of convergence in \eqref{eq:diskr} is optimal. In fact, the conjecture is proven in the one- and two-dimensional case, \cite{Sch72}, and there is theoretical and computational evidence that it is also true for higher dimensions. In practical applications, explicit examples of low-discrepancy sequences are of course needed. Among them, Sobol sequences are the most commonly used class. Since their concrete construction is not of interest for us, we refer the reader to \cite{BF88}, \cite{Gla03}, \cite{Nie92} and for an algorithm which is currently regularly used in software implementations to \cite{JK08}.
\paragraph{Koksma-Hlawka inequality.} Quasi-Monte Carlo methods are often preferred to Monte Carlo ones due to a better rate of convergence and deterministic error bounds: for an unknown function $f: [0,1]^s \to \RR$ the speed of convergence of the finite sum $\frac{1}{n} \sum_{i=1}^n f(x_i)$ to the integral $\int_{[0,1]^s} f(x) \mathrm{d} x$ depends only on the bounded variation of the function and the star-discrepancy of the sequence. More precisely the following holds.
\begin{thm}[Koksma-Hlawka inequality] \label{thm:Koksma_Hlawka} Let $f: [0,1]^s \to \RR$ be an arbitrary function of bounded variation in the sense of Hardy and Krause, $V(f)$, and let $x_1,\ldots,x_N$ be a finite sequence of points in $[0,1]^s$. Then 
	$$\left| \frac{1}{N} \sum_{i=1}^N f(x_i) - \int_{[0,1]^s} f(x) \mathrm{d} x \right| \leq V(f) D_N^*(x_1,\ldots,x_N).$$
\end{thm}
If all partial mixed derivatives of $f$ are continuous on $[0,1]^s$ then $V(f)$ can be expressed as $\sum_{u} \int_0^1 \left| \tfrac{\partial^{|u|}f}{\partial x_u}(x_u,1) \right| \mathrm{d}x_n,$ where the sum is taken over all subsets $u \subset \left\{ 1,\ldots,s\right\}$ and $(x_u,1)$ is the vector whose $i$-th component is $x_i$ if $i \in u$ and $1$ otherwise, see \cite{KN74}, Chapter~2. In contrast, a typical Monte Carlo approach would have a \textit{probabilistic} convergence rate of $\sqrt{N}$ which is much worse than the \textit{deterministic} convergence of $N^{-1} (\log N)^{s}$ implied by the Koksma-Hlawka inequality for a low-discrepancy sequence.\\[12pt]
In the LSMC context, the actual function $f$ of the BEL is indeed unknown because the cashflow projection used to calculate BEL is a complicated software which incorporates the complex interaction of liabilities, assets and management actions, compare e.g. \cite{Dof14}. Hence, whenever it comes to integration problems involving $f$, it is essential to control the star-discrepancy of $x_1,\ldots,x_N$ and use low-discrepancy sequences.

\section{Regression design}

\paragraph{Calculation of $X^TX$.} We have argued that the main numerical challenge in the implementation of \eqref{eq:estimator} lies in the inversion of $X^TX$ because the number of rows of $X$ is huge (equal to the number of outer scenarios). We now calculate the entries of $X^TX$.
Let  $\varphi_1(x_1,\ldots,x_s), \varphi_2(x_1,\ldots,x_s),\ldots$ be an arbitrary (multi-dimensional) Hilbert basis of $L^2([0,1]^s)$. The regression yields a projection to some $m$-dimensional subspace of $L^2([0,1]^s)$ and it may without loss of generality be assumed that its basis (as vector space) is given by $\varphi_1(x_1,\ldots,x_s),\ldots,\varphi_{m}(x_1,\ldots,x_s)$. Furthermore let $t^1=(t^1_1,\ldots,t^1_s),t^2=(t^2_1,\ldots,t^2_s),\ldots,t^N=(t^N_1,\ldots,t^N_s)$ be the $s$-dimensional low-discrepancy sequence used as risk driver realizations for the outer scenarios. Then
%Suppose that we have chosen a one-dimensional polynomial Hilbert basis $p_1(x),p_2(x),\ldots$ and its corresponding multi-dimensional basis given by \eqref{eq:pols}. Attribute some ordering to the multi-dimensional basis and denote it by $\varphi_1(x_1,\ldots,x_s), \varphi_2(x_1,\ldots,x_s),\ldots$. The regression yields a projection to some $n$-dimensional subspace of $L^2([0,1]^s)$ and it may without loss of generality be assumed that its basis (as vector space) is given by $\varphi_1(x_1,\ldots,x_s),\ldots,\varphi_{m}(x_1,\ldots,x_s)$. Furthermore let $y^1=(y^1_1,\ldots,y^1_s),y^2=(y^2_1,\ldots,y^2_s),\ldots,y^N=(y^N_1,\ldots,y^N_s)$ be the $s$-dimensional low-discrepancy sequence used as risk driver realizations for the outer scenarios. Then
$$X_{ij} = \varphi_{j}(t^i)$$
and hence
$$(X^TX)_{ij} = \sum_{k=1}^N \varphi_i(t^k) \varphi_j(t^k).$$

\paragraph{The main result.} Before we can formulate our main result, we need to define the expression 
$$V(\varphi)_{\max} := \max \left\{ V(\varphi_i \varphi_j) \ | \ 1 \leq i,j \leq n \right\}$$  to be the maximal Hardy-Krause variation of the pairs of basis functions appearing in $X^TX$. 
\begin{thm} \label{thm:main_thm} Let $t_1,t_2,\ldots$ be an $s$-dimensional low-discrepancy sequence with $D_N^*(Z) \leq C \tfrac{(\log N)^{s}}{N}$ for all $N \in \NN$ and let $\varphi_1,\varphi_2,\ldots,\varphi_m$ be an orthonormal basis of the $m$-dimensional subspace on which the regression projects. Furthermore, let $\theta > 1$ be arbitrary and let $N$ be such that
	$$\frac{(\log N)^{s}}{N} \leq \frac{\theta-1}{(1+\theta)CV(\varphi)m}$$
then the condition number satisfies
\begin{align} \label{eq:main}
\kappa(\tfrac{1}{N}X^TX) \leq \theta.
\end{align}
\end{thm}
Thus, if the number of outer scenarios is large, there exists an explicit upper bound for the condition number. It will follow from the proof that $\kappa(\tfrac{1}{N}X^TX)$ converges to $1$ for any uniformly distributed sequence since the latter is equivalent to $D_N^*(Z) \to 0$ for $N \to \infty$. However, the explicit bound for $N$ is only true for low-discrepancy sequences and is best possible if the answer to the great open problem of discrepancy theory is true, i.e. the best possible rate of convergence $D_N^*(Z) \to 0$ is $N^{-1} (\log N)^{s}$.
\begin{proof}
The Koksma-Hlawka inequality, Theorem~\ref{thm:Koksma_Hlawka}, and the fact that $(t_i)_{i \in \NN}$ is a low-discrepancy sequence imply that
$$ \left| (\tfrac{1}{N}X^TX)_{ij} - \int_{[0,1]^s} \varphi_i(x) \varphi_j(x) \mathrm{d} x \right| \leq V(\varphi)_{\max} C \frac{(\log N)^{s}}{N},$$
where $C$ is a constant independent of $N$. Since the basis is orthonormal, we obtain
$$\left| (\tfrac{1}{N}X^TX)_{ij} - \delta_{ij} \right| \leq V(\varphi)_{\max} C \frac{(\log N)^{s}}{N},$$
where $\delta_{ij}$ denotes Kronecker delta. This means that $\frac{1}{N}X^TX$ converges to the identity matrix for $N \to \infty$. Finally, it can be deduced from Gershgorin's Theorem~\ref{thm:Gershgorin} that for each eigenvalue $\lambda(N)$ of $\frac{1}{N}X^TX$ we have
\begin{align} \label{eq:bound:eigen}
|\lambda(N) - 1| \leq V(\varphi)_{\max} C \frac{(\log N)^{s}}{N} m.
\end{align}
Now let $0<r<1$ be arbitrary and choose $N$ large enough such that
\begin{align} \label{eq:bound:N}
\frac{(\log N)^{s}}{N} < \frac{r}{C V(\varphi)_{\max} m}.
\end{align}
Then it follows from \eqref{eq:bound:eigen} and \eqref{eq:bound:N} that
\begin{align*}
|\lambda(N) - 1| \leq V(\varphi)_{\max} C \frac{(\log N)^{s}}{N} m < r,
\end{align*}
i.e. all eigenvalues lie in the interval $(1-r,1+r)$. Finally, by \eqref{eq:kappa2} we get
\begin{align} \label{eq:kappa}
\kappa(\tfrac{1}{N}X^TX) \leq \frac{1+r}{1-r}.
\end{align}
If $\tfrac{1+r}{1-r} \leq \theta$ or in other words $r \leq \tfrac{\theta-1}{\theta+1}$, the claim follows.
\end{proof}

\paragraph{Summary and numerical results.} We have just derived a bound for the number of outer scenarios necessary to guarantee numerical stability of the LSMC regression model. Ceteris paribus, the expression $\kappa(\tfrac{1}{N}X^TX) - 1$ is supposed to be smaller than $N^{-1}(\log N)^s$ times some constant. This was confirmed by our numerical calculation in dimension $1$ using MATLAB: for the van der Corput sequence in base $2$ and shifted Legendre polynomials up to degree $2$, i.e. $m=3$, Figure~1 shows that the quotient of $\kappa(\tfrac{1}{n}X^TX) - 1$ by $N^{-1}(\log N)^s$ is clearly bounded as is predicted by Theorem~\ref{thm:main_thm} although there is some variance in the expression.
\begin{center}
	\includegraphics[scale = 0.57]{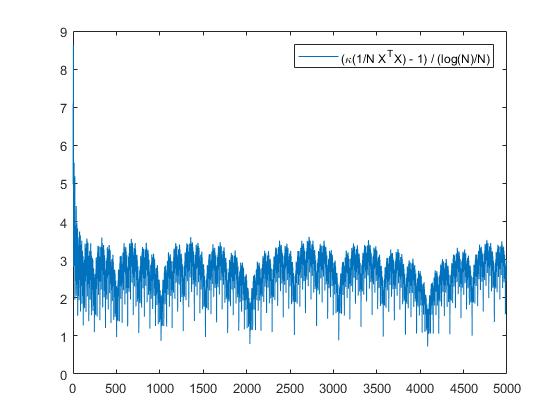}\\
	Figure 1. Quotient of $\kappa(\tfrac{1}{n}X^TX) - 1$ by $N^{-1}(\log N)^s$ for $N=10,\ldots,5000$.\\[12pt]
\end{center}
Moreover, the number of outer scenarios $N$ has to be chosen according to \eqref{eq:bound:N} in order to achieve a desired numerical precision of the LSMC calculation. It depends on four different variables:
\begin{itemize}
	\item the dimension of the (polynomial) subspace $m$: \cite{BH15} argued that a large dimension of the subspace slows down the speed of convergence. Similarly, we see that it has also a negative influence on numerical stability. Nevertheless, there may be external needs, e.g. a complex interaction of the different risk drivers, which demand for a certain cardinality of the basis.
	\item the Hardy-Krause variation of the chosen orthonormal basis $V(\varphi)_{\max}$: this shows that the choice of the orthonormal basis, which can be made by the user, has a significant influence on the numerical stability of the regression problem.
\end{itemize}
Note that there is an interaction of $m$ on $V(\varphi)_{\max}$ since considering an additional basis element might (and usually will) increase $V(\varphi)_{\max}$. We also did a numerical calculation confirming this observation (van der Corput sequence in base $2$, Legendre polynomials with maximal degree $m-1$ and $N=200$). If we only look at the quotient of $\kappa(\tfrac{1}{N}X^TX)$ by $m$, Figure~2 shows that it is clearly not bounded. Nevertheless, Figure~3 indicates that the maximal Hardy-Krause variation $V(\varphi)_{\max}$ grows faster than this quotient which is again consistent with the theoretical prediction. 
\begin{center}
	\includegraphics[scale = 0.57]{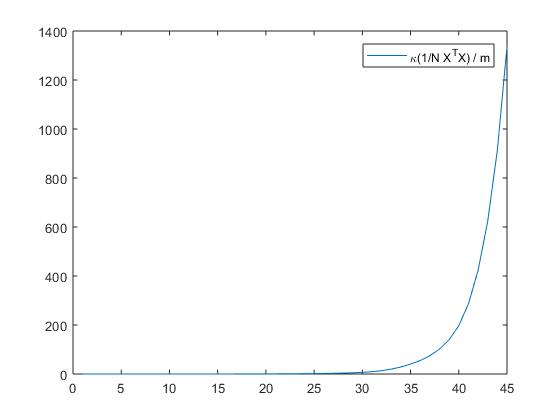}\\
	Figure 2. Quotient of $\kappa(\tfrac{1}{n}X^TX) - 1$ by $N^{-1}(\log N)^s$ for $m=1,\ldots,45$.\\[12pt]
\end{center}
\begin{center}
	\includegraphics[scale = 0.57]{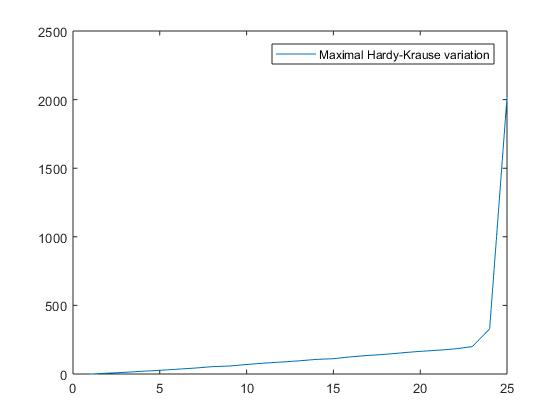}\\
	Figure 3. $V(\varphi)_{\max}$  for $m=1,\ldots,25$.\\[12pt]
\end{center}
%\todo[inline, fancyline, color=red!24]{Eine praktische Frage: Hat man überhaupt eine Chance, die Hardy-Krause-Variation in den Anwendungen auszurechnen? Für das theoretische Ergebnis ist dies nicht relevant, für die Praxis stelle ich mir die Berechnung als praktisch unmöglich vor. Was denkst du?}
We only calculated $V(\varphi)_{\max}$ up to $m=25$ here because it is numerically very challenging to give precise values and $V(\varphi)_{\max}$ grows very fast; it is already $>10^{18}$ for $m=30$. Another factor with an impact on the necessary number of simulations is: 
\begin{itemize} 
	\item the convergence constant of the low-discrepancy sequence $C$: it is a well-known phenomenon that the speed of convergence of different low-discrepancy sequences to uniform distribution differs, see e.g. \cite{Nie92}, Theorem~3.6.
\end{itemize} 
Figure~4 shows that the speed of convergence of $\kappa(\tfrac{1}{n}X^TX) - 1 \to 0$ for $N \to \infty$ depends on the chosen low-discrepancy sequence (shifted Legendre polynomials up to degree $2$, i.e. $m=3$, and van der Corput sequence in base $b$). Theoretically, it is expected that $C$ grows with increasing $b$, see again \cite{Nie92}, Theorem 3.5. This behavior is reflected by Figure~4. 
\begin{center}
	\includegraphics[scale = 0.57]{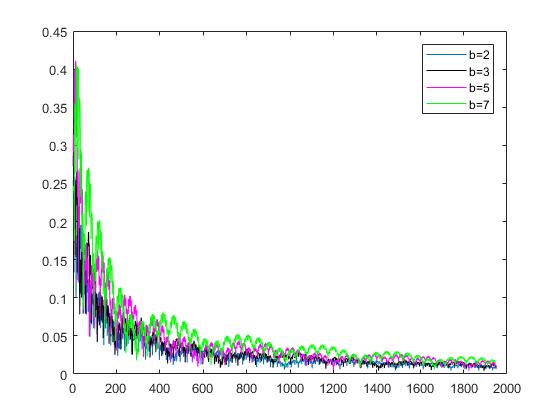}\\
	Figure 4. Convergence of $\kappa(\tfrac{1}{n}X^TX) - 1$ for different bases $b$.\\[12pt]
\end{center}
Much more delicate than choosing a good one-dimensional low-discrepancy sequence (e.g. by adjusting the base of the van der Corput sequence) is the dependence of $C$ on the dimension of the fitting space $s$, compare again e.g. \cite{Nie92}, Theorem~3.6. Therefore, a wise choice of the (Sobol) sequence used to fill the fitting space $[0,1]^s$ is of high importance for numerical stability. This leads us to the variable with the most complex interaction with the condition number:
\begin{itemize}
	\item the dimension of the fitting space $s$: it is governed by the number of risk drivers which were identified to be relevant for the insurance company. Therefore, it is given externally and cannot be influenced by a smart design of the regression algorithm. 
\end{itemize}
On the one hand, the dimension $s$ has a direct impact on the necessary number of simulations by the power of $(\log N)^s$. As $(\log N)$ grows relatively slowly, this effect is less important than the indirect implications of $s$: as we have just discussed the convergence constant of the low-discrepancy sequence $C$ grows with $s$. From a practical point of view, it also would not make sense to keep the number of basis functions $m$ constant when increasing $s$ since every relevant risk driver should be reflected by the regression algorithm. This finally also implies a possible increase of $V(\varphi)_{\max}$.\\[12pt]
Summarized, Theorem~\ref{thm:main_thm} shows that a sound design of an LSMC regression algorithm should at least incorporate orthonormal polynomials and low-discrepancy sequences. We leave it as a problem for future research to underpin our theoretical observation by further explicit numerical calculations and, by that, to find an optimal combination of the involved variables. 
%\textcolor{red}{\paragraph{Authors contributions.} The first-named author contributed the original idea in the main result of this work.}

\paragraph{Acknowledgment.} The first-named author conducted parts of the work on this paper during a stay at the Fields Institute whom he would like to thank for hospitality. Both authors would like to thank the referee for his valuable comments.

Christian Wei\ss\\
\textsc{Hochschule Ruhr West, Duisburger Str. 100, D-45479 M\"ulheim an der Ruhr, Germany}\\
\textit{E-mail address:} \texttt{christian.weiss@hs-ruhrwest.de}\\[12pt]
Zoran Nikoli\'{c}\\
\textsc{Mathematical Institute, University Cologne, Weyertal 86-90, 50931 Cologne, Germany}\\
\textit{E-mail address:} \texttt{znikolic@uni-koeln.de}


\begin{thebibliography}{xxx}
	\bibitem[AW85]{AW85} Askey, R., Wilson, J.: ``Some basic hypergeometric orthogonal polynomials that generalize Jacobi polynomials'', Memoirs of the AMS, 54 (319), iv+55 (1985).
	\bibitem[BH15]{BH15} Bauer, D., Ha, H.:  ``A least-squares Monte Carlo approach to the calculation of capital requirements'', Working Paper (2015).
	\bibitem[BBR10]{BBR10} Bauer, D., Bergmann, D., Reuss, A.: ``Solvency II and Nested Simulations -- a Least Squares Monte Carlo Approach'', International Congress of Actuaries (2010).
	\bibitem[Ben17]{Ben17} Benedetti, G.: ``A Note on the Calculation of Risk Measures Using Least Squares Monte Carlo'', International Journal of Theoretical and Applied Finance, 20 (2017).
	\bibitem[BFW14]{BFW14} Bettels, C., Fabrega, J., Wei\ss, C.: ``Anwendung von Least Squares Monte Carlo im Solvency II Kontext -- Teil 1'', Der Aktuar, 2014 (2), 151--155 (2014).
	\bibitem[BF88]{BF88} Bratley, P., Fox, B. L.: ``Algorithm 659: Implementing Sobol's quasirandom sequence generator'', ACM Trans. Math. Software, 14, 88--100 (1988).
	\bibitem[Car96]{Car96} Carriere, J.: ``Valuation of the early-exercise price for options using simulations and nonparametric
	regression'', Insurance: Mathematics and Economics 19, 19--30 (1996).
	\bibitem[Dof14]{Dof14} Doff, R.: ``The Solvency II Handbook: Practical Approaches to Implementation'', Risk Books (2014).
	\bibitem[Ger31]{Ger31} Gerschgorin, S: ``\"Uber die Abgrenzung der Eigenwerte einer Matrix'', Izv. Akad. Nauk. USSR Otd. Fiz.-Mat. Nauk (in German), 6: 749--754 (1931).
	\bibitem[Gla03]{Gla03} Glasserman, P.: ``Monte Carlo Methods in Financial Engineering'', Springer (2003).
	\bibitem[GJ10]{GJ10} Gordy, M.B., Juneja, D.: ``Nested simulations in portfolio risk measurement'', Management Science, 56, 1833--1848 (2010).
	\bibitem[JK08]{JK08} Joe, S., Kuo, F.: ``Constructing Sobol' Sequences with Better Two-Dimensional Projections'', Siam J. Sci. Comput., 30(5), 2635--2654 (2008).
	\bibitem[KNK18]{KNK18} Krah, A.-S., Nikoli\'{c}, Z., Korn. R.: ``A Least-Squares Monte Carlo Framework in Proxy Modeling of Life Insurance Companies'', Risks 6(2), (2018).
	\bibitem[KN74]{KN74} Kuipers, L., Niederreiter, H.: ``Uniform distribution of sequences'', John Wiley \& Sons, New York (1974).
	\bibitem[LHKB13]{LHKB13} Leitschkis, M., H\"orig, M., Ketterer, F., Bettels, C.: ``Least Squares Monte Carlo for Fast and Robust Capital Projections'', Milliman White Paper (2013).
	\bibitem[LS01]{LS01} Longstaff, F., Schwarz, E.: ``Valuing American options by simulation: Asimple least-squares
	approach'', The Review of Financial Studies 14: 113--147 (2001).
	\bibitem[New97]{New97} Newey,  W.K.: ``Convergence rates and asymptotic normality for series estimators'',Journal of Econometrics, 79,  147--168 (1997).
	\bibitem[Nie92]{Nie92} Niederreiter, H.: ``Random Number Generation and Quasi-Monte Carlo Methods'', Number 63 in CBMS-NSF Series in Applied Mathematics, SIAM, Philadelphia (1992).
	\bibitem[NJZ17]{NJZ17} Nikoli\'{c}, Z, Jonen, C., Zhu, C.:  ``Robust Regression Technique in LSMC Proxy Modeling'', Der Aktuar, 2017 (1), 8--16 (2017).
	\bibitem[Sch72]{Sch72} Schmidt, W. M.: ``Irregularities of distribution VII'', Acta Arith., 21, 45--50 (1972). 
	\bibitem[TGSY95]{TGSY95} Tikhonov, A., Goncharsky, A., Stepanov, V., Yagola, A.: ``Numerical Methods for the Solution of Ill-Posed Problems'', Mathematics and its Applications, Springer (1995).
	\bibitem[TB97]{TB97} Trefethen, L., Bau, D.: ``Numerical Linear Algebra'', Siam (1997).
	
\end{thebibliography}
\end{document}